\documentclass[11pt,reqno]{amsart}

\usepackage{amsmath,amsthm,amssymb,comment,fullpage}
\usepackage{braket}
\usepackage{mathtools}

\usepackage{caption}
\usepackage{centernot}
\usepackage{times}
\usepackage[T1]{fontenc}
\usepackage{mathrsfs}
\usepackage{blindtext}
\usepackage{hyperref}
\usepackage{latexsym}
\usepackage[dvips]{graphics}
\usepackage{epsfig}
\usepackage{svg}
\usepackage{amsmath,amsfonts,amsthm,amssymb,amscd}
\input amssym.def
\input amssym.tex
\usepackage{color}
\usepackage{hyperref}
\usepackage{url}
\usepackage{cases}
\newcommand{\bburl}[1]{\textcolor{blue}{\url{#1}}}

\usepackage{tikz}
\usepackage{tkz-tab}
\usetikzlibrary{shapes.geometric,positioning}

\newcommand{\burl}[1]{\textcolor{blue}{\url{#1}}}

\newcommand{\abs}[1]{\left|#1\right|}

\numberwithin{equation}{section}

\newtheorem{thm}{Theorem}[section]

\newtheorem{lem}[thm]{Lemma}
\newtheorem{prop}[thm]{Proposition}

\theoremstyle{plain}

\newtheorem{theorem}[thm]{Theorem}

\newcommand\be{\begin{equation}}
\newcommand\ee{\end{equation}}
\newcommand\bee{\begin{equation*}}
\newcommand\eee{\end{equation*}}
\newcommand\bea{\begin{eqnarray}}
\newcommand\eea{\end{eqnarray}}
\newcommand\beae{\begin{eqnarray*}}
\newcommand\eeae{\end{eqnarray*}}
\newcommand\bi{\begin{itemize}}
\newcommand\ei{\end{itemize}}
\newcommand\ben{\begin{enumerate}}
\newcommand\een{\end{enumerate}}
\newcommand\bc{\begin{center}}
\newcommand\ec{\end{center}}
\newcommand\ba{\begin{array}}
\newcommand\ea{\end{array}}

\newcommand\frakfamily{\usefont{U}{yfrak}{m}{n}}
\DeclareTextFontCommand{\textfrak}{\frakfamily}

\newcommand{\hr}[1]{\href{#1}{\url{#1}}}

\title{Distinct Angles in General Position}

\author{Henry L. Fleischmann}
\email{\textcolor{blue}{\href{mailto:henryfl@umich.edu}{henryfl@umich.edu}}}
\address{Department of Mathematics, University of Michigan, Ann Arbor, 48109}

\author{Sergei V. Konyagin}
\email{\textcolor{blue}{\href{mailto:konyagin23@gmail.com}{konyagin23@gmail.com}}}
\address{Steklov Institute of Mathematics, 8 Gubkin Street, Moscow, 119991, Russia}

\author{Steven J. Miller}
\email{\textcolor{blue}{\href{mailto:sjm1@williams.edu}{sjm1@williams.edu}},  \textcolor{blue}{\href{Steven.Miller.MC.96@aya.yale.edu}{Steven.Miller.MC.96@aya.yale.edu}}}
\address{Department of Mathematics and Statistics, Williams College, Williamstown, MA 01267}

\author{Eyvindur A. Palsson}
\email{\textcolor{blue}{\href{mailto:palsson@vt.edu}{palsson@vt.edu}}}
\address{Department of Mathematics, Virginia Tech, Blacksburg, VA 24061}

\author{Ethan Pesikoff}
\email{\textcolor{blue}{\href{mailto:ethan.pesikoff@yale.edu}{ethan.pesikoff@yale.edu}}}
\address{Department of Mathematics, Yale University, New Haven, CT 06511}

\author{Charles Wolf}
\email{\textcolor{blue}{\href{mailto:charles.wolf@rochester.edu}{charles.wolf@rochester.edu}}}
\address{Department of Mathematics, Rochester, NY, 14627}

\thanks{The preparation of this paper was only possible due to the Twentieth Annual Workshop in
Combinatorial and Additive Number Theory  (CANT 2022,  May 24--27, 2022). The authors
are immensely grateful to Melvyn Nathanson for organizing this conference. In addition, this work was supported by NSF grant $\#$1947438 and Williams College. E. A. Palsson was supported in part by Simons Foundation grant $\#$360560. }



\date{\today}

\begin{document}

\maketitle

\begin{abstract} 
The Erd\H{o}s distinct distance problem is a ubiquitous problem in discrete geometry. Somewhat less well known is Erd\H{o}s' distinct angle problem, the problem of finding the minimum number of distinct angles between $n$ non-collinear points in the plane. Recent work has introduced bounds on a wide array of variants of this problem, inspired by similar variants in the distance setting. 

In this short note, we improve the best known upper bound for the minimum number of distinct angles formed by $n$ points in general position from  $O(n^{\log_2(7)})$ to $O(n^2)$. Before this work, similar bounds relied on projections onto a generic plane from higher dimensional space. In this paper, we employ the geometric properties of a logarithmic spiral, sidestepping the need for a projection.

We also apply this configuration to reduce the upper bound on the largest integer such that any set of $n$ points in general position has a subset of that size with all distinct angles. This bound is decreased from $O(n^{\log_2(7)/3})$ to $O(n^{1/2})$.
\end{abstract}

\tableofcontents

\section{Introduction}

In 1946, Erdős introduced the distinct distance problem in his paper ``On sets of distances of $n$ points," conjecturing that the minimum number of distinct distances formed by $n$ points in the plane was $\Theta(n/\sqrt{\log n})$, the number of distances formed by points in $\sqrt{n} \times \sqrt{n}$ integer lattice \cite{ErOg}. This problem, while simple to state, proved challenging. In 2015, Guth and Katz finally proved a nearly matching lower bound of $\Omega(n/\log n)$ on the minimal number of distinct distances \cite{GuthKatz}. Since 1946, numerous variants of the problem have been considered, including the minimum number of distinct distances on restricted point sets. 

There is an analogous, far less studied problem for angles introduced by Erdős and Purdy \cite{ErPur}. What is $A(n)$, the minimum number of distinct angles formed by $n$ not all collinear points on the plane? In the angle setting, regular $n$-gon's are conjectured optimal. 

Recent work introduces new bounds on a variety of variants of the distinct angle problem \cite{GenPaper}. In particular, $A_{\textrm{gen}}(n)$, the minimum number of distinct angles formed by $n$ points in general position (with no three points on a line and no four on a circle) is shown to be $\Omega(n)$ and $O(n^{\log_2(7)})$. In this paper, we show first show that the methods in \cite{GenPaper} can be extended to provide a bound of $O(n^2 2^{O(\sqrt{\log n}}))$. We discuss this proof in Section \ref{sec: discussion of methods}. Then by an altogether new method which avoids projections altogether and chooses a configuration of points on a logarithmic spiral, we have the following.

\begin{theorem}\label{thm:AGenN^2}
We have $A_{\textrm{gen}}(n) = O(n^2)$.
\end{theorem}

Theorem \ref{thm:AGenN^2} is proved in Section \ref{sec: gen position bound}.

In Section \ref{sec: Rgen bound} we consider a related variant of this distinct angle problem also considered in \cite{GenPaper}. This is the minimum maximum size of a subset of $n$ points in general position yielding all distinct angles, $R_{\textrm{gen}}(n)$.  In other words, $R_{\textrm{gen}}(n)$ is the largest integer such that any planar point-set of $n$ points contains a subset of the given size without repeated angles. In \cite{GenPaper} $R_{\textrm{gen}}(n)$ is shown to be $O(n^{\log_2(7)/3}$ and $\Omega(n^{1/5})$. As an application of the logarithmic spiral configuration we show the following.

\begin{theorem}\label{thm: distinct angle set bd}
We have $R_{\textrm{gen}}(n) = O(\sqrt{n})$.
\end{theorem}

\section{Discussion of Methods} \label{sec: discussion of methods}
In \cite{GenPaper}, the bound $A_{\textrm{gen}}(n) = O(n^{\log_2(7)})$ is proved by projecting the points of a $d$-dimensional hypercube onto a generic plane. The argument relies closely on an observation from a paper of Erd\H{o}s, Hickerson, and Pach \cite{EHP}. Given an orthogonal projection $T$ and points $p_1, p_2, p_3$, and $p_4$, 
\begin{equation}\label{eqtn: projection-property}
    p_1 - p_2 \ = \ p_3 - p_4 \ \implies \  d(T(p_1), T(p_2)) \ = \ d(T(p_3), T(p_4)).
\end{equation}
This follows from orthogonal projections being idempotent and self-adjoint. In \cite{GenPaper}, this observation is extended. Two (congruent) triangles with edges composed of the same difference vectors are mapped to congruent triangles under orthogonal projections. Hence, it suffices to count the number of classes of translation equivalent triangles to asymptotically bound the number of distinct angles in the configuration.

It turns out that a similar argument can be used to show that $A_{\textrm{gen}}(n) = O(n^2 2^{O(\sqrt{\log n})})$. Orthogonal projections can easily be chosen such that no four points in the projection lie on a circle. However, since the projection must be injective, points on a line are projected onto a line. Hence, the original configuration must have no three points on a line.

In \cite{GenPaper} this is avoided by drawing the points from a hypercube. However, in the paper of Erdős, Füredi,  Pach, and Ruzsa showing the best known bound for the distance problem in general position, the points are instead drawn from a lattice \cite{EFPA}. The potential obstruction of three points on a line is avoided by taking a subset of the lattice points intersecting with a hypersphere. We outline a similar argument below to get an improved bound to illustrate how this projection technique may be extended.  We take inspiration from a paper of Behrend \cite{Be}.

\begin{prop}
We have $A_{\textrm{gen}}(n) = O(n^2 2^{22\sqrt{\log_2 n}})$.

\end{prop}
\begin{proof}
Consider a grid
$G_{r,d} = \{0,\dots,r\}^d$.\\

The triples $(a,b,c)$ and $(a',b',c')$ are equivalent if the second triple can be obtained from the first triple by translation. If we have any triple $(a,b,c)$, then for $i=1,\dots,d$ we can replace the triple of integers $(a_i,b_i,c_i)$ by $(a_i-m_i, b_i-m_i,c_i-m_i)$ where $m_i = \min(a_i,b_i,c_i)$. If we do this for all $i$, we get an equivalent triple $(a',b',c')$ satisfying $\min(a_i',b_i', c_i')=0$ for all $i$. The number of triples $(a_i',b_i',c_i')$ with $a_i',b_i',c_i'\in\{0,\dots,r\}$ and $\min(a_i',b_i',c_i')=0$ is $(r+1)^3 - r^3$. Thus, the number of reduced triples $(a',b', c')$ is $N_{r,d} = ((r + 1)^3 - r^3)^d$. Hence, the number of angles formed by points from $G_{r,d}$ is at most $N_{r,d}/2$, since our triples are ordered.\\

For $r>1$ the points in $G_{r,d}$ are not in general position: there are many lines containing three or more points. For $a\in G_{r,d}$ we define $f(a) = \sum_{i=1}^d a_i^2$. We have $0\le f(a) \le dr^2$. For $l=0,\dots, dr^2$ we define $G_{r,d,l} = \{a\in G_{r,d}: f(a) = l\}$. We can take $l$ so that $|G_{r,d,l}| \geq (r+1)^d (dr^2+1)^{-1},$ this quantity being the mean of the number of points at each radius $0,\dots,dr^2$. No three points from $G_{r,d,l}$ are on a line, as they lie on a common sphere. Taking a subset of the points of $G_{r,d,l}$, there is a set of $M \coloneqq (r+1)^d (dr^2 + 1 )^{-1}$ points with no three on a line. 

Now, let $r = 2^d$ and assume for simplicity that $M = \left \lfloor 2^{d(d-2)}/d \right \rfloor$. For large enough $n$, there exists $d$ such that $2^{(d-1)(d-3)}/(d-1) < n \leq M$. Then, from the above, there exists some $l$ such that a subset of  $G_{r,d,l}$ has $n$ points. This subset has no points on a line, so the configuration can be projected onto a planar configuration in general position. So, it suffices to bound the number of translation equivalence triples by $N_{r,d}$ to yield a bound on $A_{\textrm{gen}}(n)$. 

Now, note that, for $d \geq 17$, $d^2 \geq 16d + 4\log_2 d$. Then, 
\begin{align*}
dn \ &\geq \ 2^{(d-1)(d-3)} \implies \\
\log_2 n \  &\geq \ (d-1)(d-3) - \log_2 d \implies \\
4\log_2 n \ &\geq \ 3d^2 + d^2 - 16d + 12 - 4\log_2 d \geq d^2 \implies \\
2\sqrt{\log_2 n} \ &\geq \  d.
\end{align*}

Now, we have
\begin{align*}
    N_{r,d} \ = \ (3r^2 + 3r + 1)^d \ &\leq \ (4r^2)^d \\
    &= \ (2^{2d +2})^d \\
    &= \ 2^{2(d+1)d} \\
    &< \ n^2 2^{11d} \\
    &\leq \ n^2 2^{22\sqrt{\log_2 n}},
\end{align*} 
yielding the desired result.

\end{proof}

\section{An improved bound on $A_{\textrm{gen}}(n)$} \label{sec: gen position bound}
In the previous section, the extra factor of $2^{O(\sqrt{\log n})}$ arises from taking a subset of the lattice without three points on a line. We can remove such a factor by avoiding projections altogether. In this section, we describe a configuration of points on a logarithmic spiral yielding $A_{\textrm{gen}}(n) = O(n^2)$.

Let the logarithmic spiral $S$ be given by the polar equation $r = e^{\theta}$ for $\theta \in (-\infty, \infty)$. Note that there is a group of mappings $S \to S$ given by 
\[
F_\alpha(r, \theta) \ = \ (e^{\alpha}r, \theta + \alpha).
\]
Scaling by $e^{\alpha}$ is a dilation, mapping triangles to similar triangles, as does rotation by $\alpha$. Hence, mapping via an $F_\alpha$ preserves angles.\\

We now prove Theorem \ref{thm:AGenN^2} that $A_{\textrm{gen}}(n)=O(n^2)$.
\begin{proof}
Let $S$ be given by the polar equation $r = e^{\theta}$ for $\theta \in (-\infty, \infty)$. Then, consider the collection of points $\mathcal{P} = \{(e^{j \beta}, j\beta) \,:\, j \in [n] \}$ on $S$. First, note that, for sufficiently small $\beta$, $\mathcal{P}$ lies within a small arc $S'$ of $S$.  As this arc $S'$ forms part of the boundary of its own convex hull $C$, any line $\ell$ intersecting $C$ has at most two intersections with $S'$.  Consequently no three $p\in\mathcal{P}$ lie on a common line. Likewise, since the curvature of $S$ is strictly monotone, $\beta$ can be chosen small enough such that no four points of $\mathcal{P}$ are on a common circle.

Now we show that the number of distinct angles formed by the points in $\mathcal{P}$, $A(\mathcal{P})$, is at most $3\binom{n-1}{2}$. Given a triple of distinct points $t=((e^{j_1 \beta}, j_1\beta), (e^{j_2 \beta}, j_2\beta), (e^{j_3 \beta}, j_3\beta)) \in \mathcal{P}^3$, let $m = \min\{j_1, j_2, j_3\}$. Then, the map $f_t\coloneqq F_{(1 - m)\beta}$ maps this triple to another forming the same angles, now with one of the points as $(e^\beta, \beta)$. 

Hence, each of the distinct angles formed by points in $\mathcal{P}$ is formed by a triple with one point $(e^\beta, \beta)$. Since there are $\binom{n-1}{2}$ ways to choose the other two points in the triple, and each triple can yield at most three distinct angles, we have
\begin{equation}
A(\mathcal{P}) \ \leq \ 3\binom{n-1}{2},    
\end{equation}
yielding $A_{\textrm{gen}}(n) = O(n^2)$, as desired.
\end{proof}

\section{An Improved Bound on $R_{\textrm{gen}}(n)$} \label{sec: Rgen bound}
The fact that this configuration introduces no three points on a line and no four on a circle yields an improved upper bound for the minimum maximum size of a subset of $n$ points in general position yielding all distinct angles, $R_{\textrm{gen}}(n)$.  The current best known upper bound on this quantity is $O(n^{\log_2(7)/3})$ from \cite{GenPaper}.\\

Letting $x_i,y_i\in[n]$ for $1\leq i\leq 3$, we say that two triples $(x_1,x_2,x_3),(y_1,y_2,y_3)$ are equivalent if $x_1-y_1=x_2-y_2=x_3-y_3$. We then have the following lemma.
\begin{lem} \label{lem:triples}
      Let $R\subseteq [n]$ such that $\abs{R}=m$.  If $m \choose 2$ $\geq 2n-1$, then $R$ contains a pair of distinct but equivalent triples.
\end{lem} 
\begin{proof}
    The number of pairs $(x,y)\in R^2$ such that $x>y$ is $m \choose 2$, and the maximum number of possible differences is $n-1$ (ranging from $1$ to $n-1$).  Then the condition $m \choose 2$ $\geq 2n-1$ ensures by the pigeonhole principle that there are three pairs with the same difference and hence a pair of equivalent triples. 
\end{proof}
We now prove Theorem \ref{thm: distinct angle set bd}.
\begin{proof}
Let $\mathcal{P}$ be the logarithmic spiral point configuration as in Theorem \ref{thm:AGenN^2}. Let $\mathcal{P}'\subseteq\mathcal{P}$ with $\abs{\mathcal{P}'}=m$, again assuming $m \choose 2$ $\geq 2n-1$.  Define $Q\subseteq[n]$ such that $\mathcal{P}'=\{(e^{j \beta}, j\beta) \,:\, j \in Q \}$, and, by Lemma \ref{lem:triples}, $Q$ contains a pair of equivalent triples $s = (x_1, x_2, x_3)$ and $t = (y_1, y_2, y_3)$. Therefore the triples of points in $\mathcal{P}'$ corresponding to $s$ and $t$ define repeated angles. This is because the triple of points corresponding to $s$ are mapped to those corresponding to $t$ by $F_{((y_1 - x_1)\beta)}$.

Now, note that $m \geq 2n^{1/2} + 1/2$ implies $\binom{m}{2} \geq 2n - 1$. Then $R_{\textrm{gen}}(n) = O(\sqrt{n})$, as desired.
\end{proof}


\section{Future Work}
While this paper dramatically improves the state of the art upper bound for $A_{\textrm{gen}}$ to $O(n^2)$, we still only have $A_{\textrm{gen}}(n) = \Omega(n)$.  Lessening or even eliminating this gap would be a major potential target of future research. Additionally, this paper significantly improves the upper bound of $R_{\textrm{gen}}(n)$ to $O(\sqrt{n})$ from $O(n^{\log_2(7)/3})$ in \cite{GenPaper}. Nonetheless, reducing the gap with the current lower bound of $\Omega(n^{1/5})$ (also from \cite{GenPaper}) is an open problem.

The logarithmic spiral configuration may also have applications in other angle problems, such as repeated angle problems and angle chain problems appearing in the literature. For example, see Palsson, Senger, and Wolf's work on angle chains  in \cite{AngleChains}.


\end{document}